\documentclass[10pt, conference]{IEEEtran}
\usepackage{amssymb}
\usepackage{amsthm}
\usepackage{latexsym}
\usepackage{color}
\usepackage[noend]{algorithmic}
\usepackage{cite}
\usepackage{multirow}
\usepackage{graphicx}
\usepackage{url}

\graphicspath{{ALLPLOTS/}{./}}

\begin{document}
%

\title{Parallel Algorithms for Counting Triangles in Networks with Large Degrees}

\DeclareRobustCommand*{\IEEEauthorrefmark}[1]{\raisebox{0pt}[0pt][0pt]{\textsuperscript{\footnotesize\ensuremath{\ifcase#1\or *\or \dagger\or \ddagger\or%
    \mathsection\or \mathparagraph\or \|\or **\or \dagger\dagger%
    \or \ddagger\ddagger \else\textsuperscript{\expandafter\romannumeral#1}\fi}}}}



\author{
\IEEEauthorblockN{Shaikh Arifuzzaman\IEEEauthorrefmark{1}\IEEEauthorrefmark{2},
Maleq Khan\IEEEauthorrefmark{1} and Madhav Marathe\IEEEauthorrefmark{1}\IEEEauthorrefmark{2}}
\IEEEauthorblockA{
\IEEEauthorblockA{\IEEEauthorrefmark{1}Network Dynamics \& Simulation Science Laboratory, Virginia Bioinformatics Institute\\}
\IEEEauthorblockA{\IEEEauthorrefmark{2}Department of Computer Science\\}
Virginia Tech, Blacksburg, Virginia 24061 USA \\
Email: \texttt{\{sm10, maleq, mmarathe\}@vbi.vt.edu}}
}


\maketitle

\newcommand{\todo}[1]{{\textcolor{red}{\textbf{To do:} #1}}}
\newcommand{\red}[1]{{\textcolor{red}{ #1}}}
\newcommand{\fsc}{\vspace{-7pt}}  
\newcommand{\afsc}{\vspace{-5pt}}   
\newcommand{\isc}{	\setlength{\itemsep}{2pt}   
  					\setlength{\parskip}{0pt}
  					\setlength{\parsep}{0pt}
  					\setlength\itemindent{-2pt}
				}

\begin{abstract}
Finding the number of triangles in a network is
an important problem in the analysis of complex networks. The number
of triangles also has important applications in data mining. Existing
distributed memory parallel algorithms for counting triangles are
either Map-Reduce based or message passing interface (MPI) based and
work with overlapping partitions of the given network. These
algorithms are designed for very sparse networks and do not work well
when the degrees of the nodes are relatively larger. For networks with
larger degrees, Map-Reduce based algorithm generates prohibitively
large intermediate data, and in MPI based algorithms with overlapping
partitions, each partition can grow as large as the original network,
wiping out the benefit of partitioning the network.

In this paper, we present two efficient MPI-based parallel algorithms
for counting triangles in massive networks with large degrees. The
first algorithm is a space-efficient algorithm for networks that
do not fit in the main memory of a single compute node. This
algorithm divides the network into non-overlapping partitions. The
second algorithm is for the case where the main memory of each node
is large enough to contain the entire network. We observe that for
such a case, computation load can be balanced dynamically and present
a dynamic load balancing scheme which improves the performance
significantly. Both of our algorithms scale well to large networks
and to a large number of processors.

\end{abstract}






\newtheorem{theorem}{Theorem}
\newtheorem{definition}{Definition}

\section{Introduction}
\label{sec:intro}
Counting triangles in a network is a fundamental and important algorithmic problem in the analysis of complex networks, and its solution can be used in solving many other problems such as the computation of clustering coefficient, transitivity, and triangular connectivity \cite{alon-network-motifs, chu11triangle}. Existence of triangles and the resulting high clustering coefficient in a social network reflect some common theories of social science, e.g., \textit{homophily} where people become friends with those similar to themselves and \textit{triadic closure} where people who have common friends tend to be friends themselves \cite{socialMiller}. Further, triangle counting has important
applications in graph mining such as detecting spamming activity and
assessing content quality \cite{BECC}, uncovering the thematic structure of the web \cite{eckmann2002theme}, and query
planning optimization in databases \cite{YOSSEFF}.

Network is a powerful abstraction for representing underlying relations 
in large unstructured datasets. Examples include  the web graph \cite{Broder-web}, various social networks, e.g., Facebook, Twitter \cite{twitter}, collaboration networks \cite{collaborationNet}, infrastructure networks (e.g., transportation networks, telephone networks) and  biological networks \cite{bioNet}. In the present world of technological advancement, we are deluged with data from a wide range of areas such as business and finance \cite{ApteBusinessData}, computational biology \cite{ChenBioData} and social science. Many social networks have millions to billions of users \cite{tsourakakis2009doulion, chu11triangle}. This motivates the need for space efficient parallel algorithms. 

Counting triangles and related problems such as computing clustering coefficients has a rich history \cite{ALON, SCHANK, LATAPY, tsourakakis2009doulion, SURI}. Much of the earlier algorithms are mainly based on matrix multiplication and adjacency matrix representation of the network. Matrix based algorithms \cite{ALON} are not useful in the analysis of massive networks for their prohibitively large memory requirements. In the last decade many algorithms \cite{LATAPY, SCHANK, seshadhri13triangle} have been developed using adjacency list representations. Despite the fairly large volume of work addressing this problem, only recently has attention been given to the problems associated with massive networks. Several techniques can be employed to deal with such massive graphs: streaming algorithms
\cite{triangle-paradox, srikanta-cikm}, sparsification based algorithms \cite{tsourakakis2009doulion, seshadhri13triangle}, external-memory algorithms \cite{chu11triangle}, and parallel algorithms \cite{Patric-triangle, SURI, pinar-mapreduce, srikanta-cikm}. The streaming and sparsification based algorithms are approximation algorithms.  External memory algorithms can be very I/O intensive leading to a large runtime. Efficient parallel algorithms can solve such a problem of a large running time by distributing computing tasks to multiple processors. Over the last couple of years, several parallel algorithms, either shared memory or distributed memory (MapReduce or MPI) based, have been proposed.

A shared memory parallel algorithm is
proposed in \cite{srikanta-cikm} for counting triangles in a streaming setting, which is an approximation algorithm. In \cite{SURI}, two parallel algorithms for exact triangle counting using the MapReduce framework are presented. The first algorithm generates huge volumes of intermediate data, which are all possible 2-paths centered at each node. Shuffling and regrouping these 2-paths require a significantly large amount of time and memory. The second algorithm suffers from redundant counting of triangles. An improvement of the second algorithm is given in a
very recent paper \cite{park-cikm}. Although this algorithm reduces
the redundant counting to some extent, the redundancy is not entirely
eliminated. A MapReduce based parallelization of a wedge-based sampling
technique \cite{seshadhri13triangle} is proposed in \cite{pinar-mapreduce}, which is an approximation algorithm. MapReduce framework provides several advantages such as fault tolerance, abstraction of parallel
computing mechanisms, and ease of developing a quick prototype or program.  However, the overhead for doing so results in a larger runtime. On the other hand, MPI based systems provide the advantages of defining and controlling parallelism from a granular level, implementing application specific optimizations such as load
balancing, memory and message optimization.

A recent paper \cite{Patric-triangle} proposes an MPI based parallel algorithm for counting the \textit{exact} number of triangles in massive networks. The algorithm employs an overlapping partitioning scheme and a novel load balancing scheme. Although this algorithm works very well on very sparse networks, it is not suitable for networks with large degrees. The size of the partitions grows quadratically with the increase of the degrees of the nodes. As a result, for a network with large degrees, the partitions can grow as large as the whole network, wiping out the benefit of partitioning the network. 

We present two efficient MPI-based parallel algorithms for finding the exact number of triangles in networks with large degrees. The first algorithm is a space efficient algorithm for massive networks that do not fit in the memory of a single computing machine. This algorithm divides the network into non-overlapping partitions. Not only this algorithm is suitable for networks with large degrees, even for networks with smaller degrees, it can work on larger networks than that of the algorithm in \cite{Patric-triangle} as the non-overlapping partitioning scheme leads to significantly smaller partitions. The second algorithm is for the case where the memory of each machine is large enough to contain the entire network. For such a case, we present a parallel algorithm with a dynamic load balancing scheme, which improves the performance significantly. Both of our algorithms scale well to large networks and to a large number of processors.

The rest of the paper is organized as follows. The preliminary concepts, notations and datasets are briefly described in Section \ref{sec:prelim}. In Section \ref{sec:background}, we discuss some background work on counting triangles. We present our parallel algorithms in Section \ref{sec:space} and \ref{sec:dynamic} and conclude in Section \ref{sec:conc}.


\section{Preliminaries}
\label{sec:prelim}

Below are the notations, definitions, datasets, and experimental setup used in this paper.\\ \vspace{-.1in}

\textbf{Basic definitions.} 
The given network is denoted by $G(V,E)$, where $V$ and $E$ are the sets of vertices and edges, respectively, with $m = |E|$ edges and $n = |V|$ vertices labeled as $0, 1, 2, \dots, n-1$. We use the words \emph{node} and \emph{vertex} interchangeably. We assume that the input graph is undirected. If $(u,v)\in E$, we say $u$ and $v$ are neighbors of each other. The set of all neighbors of $v \in V$ is denoted by $\mathcal{N}_v$, i.e.,  $\mathcal{N}_v=\{u \in V | (u,v) \in E\}$. The degree of $v$ is $d_v = |\mathcal{N}_v|$.

A triangle is a set of three nodes $u,v,w \in V$ such that there is an edge between each pair of these three nodes, i.e., $(u,v), (v,w), (w,u) \in E$. The number of triangles containing node $v$ (in other words, triangles incident on $v$) is denoted by $T_v$. Notice that the number of triangles containing node $v$ is same as the number of edges among the neighbors of $v$, i.e., 
\vspace{-.1in}
\begin{eqnarray*}
T_v = |\left\{(u,w) \in E \mid u, w \in \mathcal{N}_v \right\}|.
\end{eqnarray*}

We use K, M and B to denote thousands, millions and billions, respectively; e.g., 1B stands for one billion.

\textbf{Datasets.} 
We used both real world and artificially generated networks. A summary of all the networks is provided in Table \ref{table:dataset}. Twitter data set is available at \cite{twitterData}, and web-BerkStan, LiveJournal and web-Google networks are at SNAP library \cite{SNAP}. Miami \cite{miamiRef} is a synthetic, but realistic, social contact network for Miami city. Note that the web-BerkStan, web-Google, LiveJournal and Twitter networks have very skewed degree distribution, i.e, some nodes have very large degrees.
Artificial network  PA($n,d$) is generated using preferential attachment (PA) model \cite{barabasi99} with $n$ nodes and average degree $d$. PA($n,d$) has power-law degree distribution, which is a very skewed distribution. Networks having some nodes with high degrees create difficulty in partitioning and balancing loads and thus give us a chance to measure the performance of our algorithms in some of the worst case scenarios.

\begin{table}
\caption{ Dataset used in our experiments}
\label{table:dataset}
\centering
    \begin{tabular}{ | l | l | l | l |}
    \hline
    {\bf Network} & {\bf Nodes} & {\bf Edges} & {\bf Source} \\ \hline
	web-Google   &  $0.88$M    &   $5.1$M     & SNAP \cite{SNAP}\\ \hline
	web-BerkStan  &  $0.69$M    &   $13$M   & SNAP \cite{SNAP}\\ \hline
	Miami         &  $2.1$M    &   $100$M     & \cite{miamiRef}\\ \hline
	LiveJournal   &  $4.8$M    &  $86$M     & SNAP \cite{SNAP}\\ \hline
	Twitter       &  $42$M   &     $2.4$B  & \cite{twitterData}\\ \hline
	PA($n,d$)     &  $n$     &     $\frac{1}{2}nd$    & Pref. Attachment \\ \hline
	\end{tabular}
\end{table}

\textbf{Computation Model.} 
We develop parallel algorithms for message passing interface (MPI) based distributed-memory parallel systems, where each processor has its own local memory. The processors do not have any shared memory; one processor cannot access the local memory of another processor, and the processors communicate via exchanging messages using MPI.

We perform our experiments using a computing cluster (Dell C6100) with 30 computing nodes and 12 processors (Intel Xeon X5670, 2.93GHz) per node. The memory per processor is 4GB, and the operating system is SLES 11.1.

\section{A Background on Counting Triangles}
\label{sec:background}

Our parallel algorithms are based on the state-of-the-art sequential algorithm for counting triangles. In this section, we describe the sequential algorithm and some background of our parallel algorithms. 

\subsection{Efficient Sequential Algorithm}

A na\"{\i}ve approach to count triangles in a graph $G(V,E)$ is to check, for all possible triples $(u,v,w)$, $u,v,w \in V$, whether $(u,v,w)$ forms a triangle; i.e., check if $(u,v), (v,w), (u,w) \in E$. There are  ${n \choose 3}$ such triples, and thus this algorithm takes $\Omega(n^3)$ time. A simple but efficient algorithm for counting triangles is: for each node $v\in V$, find the number of edges among its neighbors, i.e., the number of pairs of neighbors that complete a triangle with vertex $v$. In this method, each triangle $(u,v,w)$ is counted six times -- all six permutations of $u, v$, and $w$. The algorithms presented in \cite{LATAPY, SCHANK}  uses a total ordering $\prec$ of the nodes to avoid duplicate counts of the same triangle. Any arbitrary ordering of the nodes, e.g., ordering the nodes based on their IDs, makes sure each triangle is counted exactly once -- counts only one permutation among the six possible permutations. However, algorithms in \cite{LATAPY, SCHANK} incorporate an interesting node ordering based on the degrees of the nodes, with ties broken by node IDs, as defined as follows: $ u \prec v \iff d_u < d_v \mbox{ or } (d_u = d_v \mbox{ and } u < v)$.

These algorithms are further improved in a recent paper \cite{Patric-triangle} by a simple modification. The algorithm \cite{Patric-triangle} defines $N_x \subseteq \mathcal{N}_x $ as the set of neighbors of $x$ having a higher order $\prec$ than $x$. For an edge $(u,v)$, the algorithm stores $v$ in $N_u$ if {$u \prec v$}, and $u$ in $N_v$, otherwise, i.e., $N_v = \{u: (u,v)\in E, v \prec u \}$.  Then the triangles containing node $v$ and any $u \in N_v$ can be found by set intersection $N_u \cap N_v$. Now let, $\hat{d_v} = |N_v|$ be the effective degree of node $v$. The cost for computing $N_u \cap N_v$ requires $O(\hat{d_v} + \hat{d_u})$ time when $N_v$ and $N_u$ are sorted.
The above state-of-the-art sequential algorithm is presented in Fig. \ref{algo:serial}. Our parallel algorithms are based on this sequential algorithm.

\begin{figure}
\begin{center}
\fbox{
\begin{minipage}[c] {0.8\linewidth}
\begin{algorithmic}[1]
\FOR {each edge $(u,v)$}
    \STATE if {$u \prec v$}, store $v$ in $N_u$
    \STATE else store $u$ in $N_v$
\ENDFOR
\FOR {$v \in V$}
	\STATE sort $N_v$ in ascending order
\ENDFOR
\STATE $T \leftarrow 0$ \hspace{0.2in} \{$T$ is the count of triangles\}
\FOR {$v \in V$}
	\FOR {$u \in N_v$}
		\STATE $S \leftarrow N_v \cap N_u$
		\STATE $T \leftarrow T+|S|$
	\ENDFOR
\ENDFOR
\end{algorithmic}
\end{minipage}
}
\end{center}
\fsc
\caption{The state-of-the-art sequential algorithm for counting triangles.}
\label{algo:serial}
\afsc
\end{figure}

\subsection{Related Parallel Algorithms}

In Section \ref{sec:intro} we discussed a few parallel algorithms \cite{Patric-triangle, SURI, pinar-mapreduce} which deal with massive networks. The most relevant to our work is the parallel algorithm presented in \cite{Patric-triangle}. The algorithm \cite{Patric-triangle} divides the input graph into a set of overlapping partitions. Let $P$ be the number of processors used in the computation. A partition (subgraph) $G_i$ is constructed as follows. Set of nodes $V$ is partitioned into $P$ disjoint subsets $V_0^c, V_1^c, \dots, V_{P-1}^c$, such that, for any $j$ and $k$, $V_j^c \cap V_k^c = \emptyset$ and $\bigcup_{k} V_k^c = V$. Then, a set $V_i$ is constructed containing all nodes in $V_i^c$ and $\bigcup_{v \in V_i^c} N_v$. The partition $G_i$ is a subgraph of $G$ induced by $V_i$. Each processor $i$ works on partition $G_i$. The node set $V_i^c$ and $N_v$ for $v\in V_i^c$ constitute the disjoint (non-overlapping) portion of the partition $G_i$. The node set $V_i-V_i^c$ and $N_u$ for $u\in V_i-V_i^c$ constitute the overlapping portion of the partition $G_i$.


This scheme works very well for sparse networks; however, when the average degrees of networks increase, each partition can be very large. Assuming an average degree $\bar{d}$ of the network, the algorithm in \cite{Patric-triangle} has a space requirement of $\Omega(n\bar{d}/P)$ for storing disjoint portion of the partition. The space needed for storing the whole  partition is $\Omega(xn\bar{d}/P)$ where $1 \leq x \leq \bar{d}$ which can be very large. In many real world networks average degrees are large, e.g., Facebook users have an average of $190$ friends \cite{fb-stat}. 

We observe that even for a sparse network with small average degree, if there are few nodes with large degrees, say $O(n)$, some partitions can be very large. For example, consider a node $v$ with degree $n-1$, where $n$ is the number nodes in the network, the partition containing node $v$ will be equal to the whole network. Some real networks have very skewed degree distributions where some nodes have very large degrees.

In the first algorithm presented in this paper, we divide the input networks into non-overlapping partitions. Each partition is almost equal and has approximately $m/P$ edges, which can be significantly smaller (in some cases, as much as $\bar{d}$ times smaller) than the overlapping partition in \cite{Patric-triangle}. As a result, our algorithm can work with networks with large degrees. Since the partitions are significantly smaller, even for a network with smaller degree the algorithm can work on larger networks than \cite{Patric-triangle}. Table \ref{table:mem-comp} shows the space requirement (in MB) of our algorithm and the algorithm in \cite{Patric-triangle}. 

 

\begin{table}[!htb]
\centering
\caption{Memory usage of our algoirthm and \cite{Patric-triangle} for storing the largest partition. Number of partitions used is $100$.}
\label{table:mem-comp}
\begin{tabular}{|l|l|l|l|} \hline
\multirow{2}{*}{ \textbf{Networks}} &  \multicolumn{2}{|c|}{\textbf{Memory (MB)}} & \multirow{2}{*}{\textbf{Avg. Degree}} \\
\cline{2-3}
 & Our algo. & Algo. in \cite{Patric-triangle} &  \\ \hline
	Miami         &    $10.63$     & $36.56$ & $47.6$    \\ \hline 
	web-Google   &    $1.49$    & $5.65$  & $11.6$   \\ \hline 
	LiveJournal   &    $9.41$    & $22.15$  & $18$   \\ \hline 
	Twitter   &      $265.82$ &  $6876.25$ &  $57.14$  \\ \hline 
	PA(10M, 100)         &    $121.11$     & $2120.94$ &  $100$  \\ \hline 
    \end{tabular}
\end{table}

Now consider the case that the size of an overlapping partition is equal (or almost equal) to the whole network and each computing machine has enough space available for storing the partition, and consequently the whole network. For such a case, we observe that a dynamic load balancing scheme can make the computation even faster and present an efficient parallel algorithm with dynamic load balancing, which is significantly faster than the algorithm in \cite{Patric-triangle}. We present our parallel algorithms in the following sections.

\section{A space efficient parallel algorithm}
\label{sec:space}

In this section, we present our space-efficient parallel algorithm for counting triangles. At first, we present the overview of the algorithm. The detailed description of the algorithm follows thereafter. 

\subsection{Overview of the Algorithm}
\label{sec:space-ov}

Let $P$ be the number of processors used in the computation. The algorithm partitions the input network $G(V,E)$ into a set of $P$ subgraphs $G_i$. Informally, the subgraph $G_i$ is constructed as follows: set of nodes $V$ is partitioned into $P$ mutually disjoint subsets $V_k$s, for $0 \leq k \leq P-1$, such that, $\bigcup_{k} V_k = V$. Node set $V_i$, along with $N_v$ for all $v\in V_i$, constitutes the subgraph $G_i$ (Definition \ref{dfn:subgraph}). Each processor is assigned one such subgraph (partition) $G_i$. Now, to count triangles incident on $v \in V_i$, processor $i$ needs $N_u$ for all $u \in N_v$. If $u\in V_i$, processor $i$ counts triangles incident on $(v,u)$ by computing $N_u \cap N_v$. However, if $u \in V_j$, $j \neq i$, then $N_u$ resides in processor $j$. In such a case, processor $i$ and $j$ exchange message(s) to count triangles adjacent to edge $(v,u)$. There are several non-trivial issues involving this exchange of messages, which can crucially affect the performance of the algorithm. We devise an efficient communication scheme to reduce the communication overhead and improve the performance of the algorithm. Once all processors complete the computation associated to respective partitions, the counts from all processors are aggregated.

\begin{definition} \label{dfn:subgraph}
\textbf{A non-overlapping partition:} Given a graph $G(V, E)$ and an integer $P\geq 1$ denoting the number of partitions, a non-overlapping partition for our algorithm, denoted by $G_i(V_i', E_i')$, is a subgraph of $G$ which is constructed as follows.
\begin{itemize}
\item $V$ is partitioned into $P$ disjoint subsets $V_i$s of consecutive nodes, such that, for any $j$ and $k$, $V_j \cap V_k = \emptyset$ and $\bigcup_{k} V_k = V$
\item $V_i'=V_i \cup \{v:v\in N_u, u\in V_i\}$
\item  $E_i' = \{(u,v):u\in V_i, v\in N_u\}$
\end{itemize}
\end{definition}

The subgraphs (partitions) $G_i$s are non-overlapping-- each edge $(u,v) \in E$ resides in one and only one partition. For any $j$ and $k$, $E_j' \cap E_k' = \emptyset$ and $\bigcup_{k} E_k' = E$. The sum of space required to store all partitions equals to the space required to store the whole network.

Our algorithm exchanges two types of messages-- data message and completion notifier. A message is denoted by $\left\langle t, X \right\rangle $ where $t \in \{data, completion\}$ is the type of the message and $X$ is the actual data associated with the message. For a data message ($t=data$), $X$ refers to a neighbor list, whereas for a completion notifier ($t=completion$), the value of $X$ is disregarded. We describe the details of our algorithm in the following subsections.

\subsection{Computing Partitions}
\label{sec:space-part}

While constructing partitions $G_i$, set of nodes $V$ is partitioned into $P$ disjoint subsets $V_i$s (Definition \ref{dfn:subgraph}). How the nodes in $V$ are distributed among the sets $V_i$ for all processors $i$ crucially affect the performance of the algorithm. Ideally, the set $V$ should be partitioned in such a way that the cost for counting triangles is almost equal for all partitions. Let, $f(v)$ be the cost for counting triangles incident on a node $v\in V$(cost for executing Line 7-10 in Fig. \ref{algo:serial}). We need to compute $P$ disjoint partitions of node set $V$ such that for each partition $V_i$, 
\begin{eqnarray*}
\sum_{v \in V_i}{f(v)} \approx \frac{1}{P}\sum_{v \in V}{f(v)}.
\end{eqnarray*}
We realize that the parallel algorithm for computing balanced partitions of $V$ proposed in \cite{Patric-triangle} is applicable to our problem. The paper \cite{Patric-triangle} proposed several estimations for $f(v)$, among which $f(v) = \sum_{ u\in N_v}{(\hat{d_v} + \hat{d_u})}$ is shown experimentally as the best. Since our algorithm employs a different communication scheme for counting triangles, none of those estimations corresponds to the cost of our algorithm. Thus, we compute a different function $f(v)$ to estimate the computational cost of our algorithm more precisely (in Section \ref{sec:new_fv}). We use this function to compute balanced partitions. Once all $P$ partitions are computed, each processor is assigned one such partition.

\subsection{Counting Triangles with An Efficient Communication Scheme}
\label{sec:space-tc}

As we discussed in Section \ref{sec:space-ov}, processor $i$ stores $N_v$ for all $v \in V_i$. However, to compute triangles incident on $v \in V_i$, $N_u$ for all $u \in N_v$ are also required. Now, if $u \in V_j$, $j \neq i$, then $N_u$ resides in processor $j$. A simple approach to resolve the issue is as follows.

\textbf{(The Direct Approach)} \textit{Processor $i$ requests processor $j$ for $N_u$. Upon receiving the request, processor $j$ sends $N_u$ to processor $i$. Processor $i$ counts triangles incident on the edge $(v,u)$ by computing $N_v \cap N_u$.}

The direct approach has a high communication overhead due to exchanging redundant messages. Assume $u \in N_v$ and $u \in N_w$ for $v,w \in V_i$. Then processor $i$ sends two separate requests for $N_u$ to processor $j$ while computing triangles incident on $v$ and $w$, respectively. In response to the above requests, processor $j$ sends same message $N_u$ to processor $i$ twice. Such redundant messages increases the communication overhead leading to poor performance of the algorithm. 

One way to eliminate redundant messages is that instead of requesting $N_u$s multiple times, processor $i$ stores them in memory after fetching them for the first time. Before sending a request, processor $i$ performs a lookup into the already fetched lists of $N_u$s. A new request for $N_u$ is made only when it is not already fetched. However, the space requirement for storing all $N_u$s along with $G_i$ is same as that of storing an overlapping partition. This diminishes our original goal of a space-efficient algorithm.

Another way to eliminate message redundancy is as follows. When a neighbor list $N_u$ is fetched, processor $i$ scans through $G_i$ to find all nodes $v \in V_i$ such that $u \in N_v$. Processor $i$, then, performs all computations related to $N_u$ (i.e., $N_v \cap N_u$). Once these computations are done, the neighbor list $N_u$ is never needed again and can be discarded. However, scanning through the whole partition $G_i$ for each fetched list $N_u$ might be very expensive, which is even more expensive than the direct approach with redundant messages.

Since all of the above techniques compromise either runtime or space efficiency, we introduce another communication scheme for counting triangles involving nodes $v\in V_i$ and $u\in V_j$, such that $u \in N_v$.

\textbf{(The Surrogate Approach)} \textit{Processor $i$ sends  $N_v$ to processor $j$.  Processor $j$ scans $N_v$ to find all nodes $u\in N_v$ such that $u\in V_j$. For all such nodes $u$, processor $j$ counts triangles incident on edge $(u,v)$  by performing the operation $N_v \cap N_u$.}\\
\vspace{-.1in}

The surrogate approach eliminates the exchange of redundant messages: while counting triangles incident on a node $v \in V_i$, processor $i$ may find multiple nodes $u\in N_v$ such that $u \in V_j$. Processor $i$ sends $N_v$ to processor $j$ when such a node $u$ is encountered for the first time. Since processor $j$ counts triangles incident on edge $(u,v)$ for all such nodes $u$, processor $i$ does not send $N_v$ again to processor $j$. 

To implement the above strategy for eliminating redundant messages, processor $i$ needs to keep record of which processors it has already sent $N_v$ to, for a node $v\in V_i$. This is done using a single variable \textit{LastProc} which records the last processor a neighbor list $N_v$ is  sent to. The variable is initialized to a negative value. When processor $i$ encounters a node $u\in N_v$ such that $u\in V_j$, it checks the value of $\textit{LastProc}$.
If $\textit{LastProc} \neq j$, processor $i$ sends $N_v$ to processor $j$ and set $\textit{LastProc} = j$. Otherwise, the node $u$ is ignored, meaning it would be redundant to send $N_v$. Once all nodes $u\in N_v$ are checked, the variable $\textit{LastProc} $ is again reset to a negative value. It is easy to see that since $V_i$ is a set of consecutive nodes, and all neighbor lists $N_v$ are sorted, all nodes $u \in N_v$ such that $u\in V_j$ reside in $N_v$ in consecutive positions. Thus, using the variable $\textit{LastProc}$ redundant messages are detected correctly and eliminated without compromising any of execution and space efficiency. This capability of surrogate approach is crucial in the runtime performance of the algorithm, as shown experimentally in Section \ref{sec:space-perf}.

The pesudocode for counting triangles for an incoming message $\left\langle data, X \right\rangle$ is given in Fig. \ref{algo:space-tc-msg}.

\begin{figure}[!ht]
\begin{center}
\fbox{
\begin{minipage}[c] {0.85\linewidth}
\begin{algorithmic}[1]
	\STATE {Procedure \textsc{SurrogateCount}$(X,i):$}
	\STATE {$T \leftarrow 0$}\hspace{0.1in} //$T$ is the count of triangles
	\FOR {all $u\in X$ such that $u \in V_i$ }
			\STATE $S \leftarrow N_u \cap X$
			\STATE $T \leftarrow T+|S|$ 
	\ENDFOR
	\RETURN $T$
\end{algorithmic}
\end{minipage}
}
\end{center}
\fsc
\caption{A procedure executed by processor $i$ to count triangles for incoming message $\left\langle data, X \right\rangle $.}
\label{algo:space-tc-msg}
\afsc
\end{figure}

\subsection{Termination}

Once a processor $i$ completes the computation associated with all $v\in V_i$, it broadcasts a completion notifier $\left\langle completion, X \right\rangle $. However, it cannot terminate its execution since other processors might send it data messages $\left\langle data, X \right\rangle $ for counting triangles (as in Fig. \ref{algo:space-tc-msg}). When processor $i$ receives completion notifiers from all other processors, aggregation of counts from all processors is performed using MPI aggregation function, and the execution terminates.


The complete pseudocode of our space efficient parallel algorithm for counting triangles using surrogate approach is presented in Fig. \ref{algo:space-tc}.

 
\begin{figure}[!ht]
\begin{center}
\fbox{
\begin{minipage}[c] {0.95\linewidth}
\begin{algorithmic}[1]
\STATE {$T_i \leftarrow 0$}\hspace{0.1in} //$T_i$ is processor $i$'s count of triangles
\FOR {$v \in V_i$}
    \FOR {$u \in N_v$}
	    \IF { $u \in V_i$}
			\STATE $S \leftarrow N_v \cap N_u$
			\STATE $T_i \leftarrow T_i+|S|$ 
		\ELSIF {$u \in V_j$ }
			\STATE {\textbf{Send} $\left\langle data, N_v \right\rangle $ to proc. $j$ if not sent already}	
		\ENDIF
	\ENDFOR
	\STATE {}	
	\STATE {\textbf{Check} for incoming messages $\left\langle t, X \right\rangle $:}
	\IF { $t=data $}
			\STATE {$T_i \leftarrow T_i +$ \textsc{SurrogateCount}$(X,i)$}
		\ELSE 
			\STATE {\textbf{Increment} completion counter}	
	\ENDIF
\ENDFOR
\STATE{}
\STATE{\textbf{Broadcast} $\left\langle completion, X \right\rangle $}
\WHILE{completion counter $<$ P-1} 
	\STATE {\textbf{Check} for incoming messages $\left\langle t, X \right\rangle $:}
	\IF { $t=data $}
			\STATE {$T_i \leftarrow T_i +$ \textsc{SurrogateCount}$(X,i)$}
		\ELSE 
			\STATE {\textbf{Increment} completion counter}	
	\ENDIF
\ENDWHILE
\STATE{}
\STATE \textsc{MpiBarrier}
\STATE \textbf{Find Sum} $T \leftarrow \sum_i{T_i}$ using $\textsc{MpiReduce}$
\RETURN  $T$

\end{algorithmic}
\end{minipage}
}
\end{center}
\fsc
\caption{An algorithm for counting triangles using surrogate approach. 
Each processor $i$ executes Line 1-22. After that, they are synchronized and the aggregation is performed (Line 24-26).}
\label{algo:space-tc}
\afsc
\end{figure}

\subsection{Correctness of The Algorithm}
\label{sec:correct_spacealgo}

The correctness of our space efficient parallel algorithm is formally presented in the following theorem.

\begin{theorem} \label{thm:correct_spacealgo}
Given a graph $G=(V, E)$, our space efficient parallel algorithm correctly counts exact number of triangles in G.
\end{theorem}

\begin{proof}
Consider a triangle $(x_1, x_2, x_3)$ in $G$, and without the loss of generality, assume that $x_1 \prec x_2 \prec x_3$. By the constructions of $N_x$ (Line 1-3 in Fig. \ref{algo:serial}), we have $x_2, x_3 \in {N_{x_1}}$ and $x_3 \in N_{x_2}$. 

Now, there might be two cases as shown below.

\begin{itemize}
\item[1.] \textit{Case $x_1, x_2 \in V_i$:}\\
Nodes $x_1$ and $x_2$ are in the same partition $i$. Processor $i$ executes the loop in Line 2-6 (Fig. \ref{algo:space-tc}) with $v = x_1$ and $u = x_2$, and node $x_3$ appears in $S = N_{x_1} \cap N_{x_2}$, and the triangle $(x_1, x_2, x_3)$ is counted once. But this triangle cannot be counted for any other values of $v$ and $u$ because $x_1 \notin N_{x_2}$ and $x_1,x_2 \notin N_{x_3}$.

\item[2.] \textit{Case $x_1 \in V_i, x_2 \in V_j, i \neq j$:}\\
Nodes $x_1$ and $x_2$ are in two different partitions, $i$ and $j$, respectively, without the loss of generality. Processor $i$ attempts to count the triangle executing the loop in Line 2-6 with $v = x_1$ and $u = x_2$. However, since $x_2 \notin V_i$, processor $i$ sends $N_{x_1}$ to processor $j$ (Line 8). Processor $j$ counts the triangle while executing the loop in Line 10-12 with $X = N_{x_1}$, and node $x_3$ appears in $S = N_{x_2} \cap N_{x_1}$(Line 2 in Fig. \ref{algo:space-tc-msg}). This triangle can never be counted again in any processor, since  $x_1 \notin N_{x_2}$ and $x_1,x_2 \notin N_{x_3}$.
\end{itemize}
Thus, in both cases, each triangle in $G$ is counted once and only once. This completes our proof.
\end{proof}

\subsection{Computing An Estimation for $f(v)$}
\label{sec:new_fv}
As discussed in Section \ref{sec:space-part}, computing balanced partitions requires an estimation of the cost $f(v)$ for counting triangles incident on node $v$. We compute a new function for estimating $f(v)$ which captures the computing cost of our algorithm more precisely, as follows. 

Set of neighbors $\mathcal{N}_v$ and $N_v$ are defined in Section \ref{sec:prelim} and \ref{sec:background}, respectively, as $\mathcal{N}_v = \{u: (u,v)\in E \}$ and $N_v = \{u: (u,v)\in E, v \prec u \}$. It is easy to see, $u \in \mathcal{N}_v-N_v \Leftrightarrow v \in N_u$.

To estimate the cost for counting triangles incident on node $v\in V_i$, we consider the cost for counting triangles incident on edges $(v,u)$ such that $u\in \mathcal{N}_v$. There might be two cases: 
\begin{itemize}
\item[1.] \textit{Case $u \in \mathcal{N}_v-N_v$:}\\
This case implies $v \in N_u$. There might be two sub-cases:
\begin{itemize}
\item If $u \in V_j$ for $j\neq i$, processor $j$ sends $N_u$ to processor $i$, and processor $i$ counts triangle by computing $N_u \cap N_v$ (Fig. \ref{algo:space-tc-msg}).
\item If $u \in V_i$, processor $i$ counts triangle by computing $N_u \cap N_v$ while executing the loop in Line 2-6 in Fig. \ref{algo:space-tc} for node $u$.
\end{itemize}
Thus for both sub-cases processor $i$ computes triangles incident on $(v,u)$. All such nodes $u$ impose a computation cost of $\sum_{ u\in \mathcal{N}_v-N_v}{(\hat{d_v} + \hat{d_u})}$ on processor $i$ for node $v$.
\item[2.] \textit{Case $u \in N_v$:}\\
This case implies $v \in \mathcal{N}_u-N_u$ which is same as case 1 with $u$ and $v$ interchanged. By a similar argument of case 1, the imposed computation cost for such $(v,u)$ is attributed to node $u$.



\end{itemize}

Thus the estimated cost attributed to node $v$ for counting triangles on all edges $(v,u)$, such that $u\in \mathcal{N}_v$, is $\sum_{ u\in \mathcal{N}_v-N_v}{(\hat{d_v} + \hat{d_u})}$.
This gives us the intended function $f(v) = \sum_{ u\in \mathcal{N}_v-N_v}{(\hat{d_v} + \hat{d_u})}$ which we use in our partitioning step. We present an experimental evaluation comparing the best estimation function presented in \cite{Patric-triangle} with ours in Section \ref{sec:space-perf}.



\subsection{Runtime and Space Complexity}
\label{sec:runtime_spacealgo}

The runtime and space complexity of our algorithm are presented below.

\paragraph*{Runtime Complexity}
The runtime complexity of our algorithm is the sum of costs for computing partitions, counting triangles, exchanging $N_v$s, and aggregating counts from all processors. Computing balanced partition takes $O(m/P + P \log P)$ time using the scheme presented in \cite{Patric-triangle}.  For the algorithm given in Fig. \ref{algo:space-tc}, the worst case cost for counting triangles is $O(\sum_{v\in V_i}\sum_{u\in \mathcal{N}_v-N_v}{(\hat{d}_u + \hat{d}_v)})$. Further, the communication cost incurred on a processor is $O(m/P)$ in the worst case. The aggregation of counts require $O(\log P)$ time using MPI aggregation function. Thus, the time complexity of our parallel algorithm is,
\begin{eqnarray*}
 O(m/P + P\log P + \max_i \sum_{v\in V_i}\sum_{u\in \mathcal{N}_v-N_v}{(\hat{d}_u + \hat{d}_v)}) 
\end{eqnarray*} 

\paragraph*{Space Complexity}
The space complexity of our algorithm depends on the size of the partitions and messages exchanged among processors. The size of the partition is $O(\max_i \{|V_i'| + |E_i'|\})$. Now our algorithm stores only a single incoming or outgoing message at a time, and each data message contains $N_v$. Thus, the space requirement for storing messages is $O(\max_{v\in V}|N_v|) = O(d_{max})$. Taking those two into account, the overall space complexity is $O(\max_i \{|V_i'| + E_i'\} + d_{max})$.

\begin{figure*}[!tbh]
\hfill
\begin{minipage}[t]{.32\textwidth}
\begin{center}
\centerline{\includegraphics[width=1.0\textwidth]{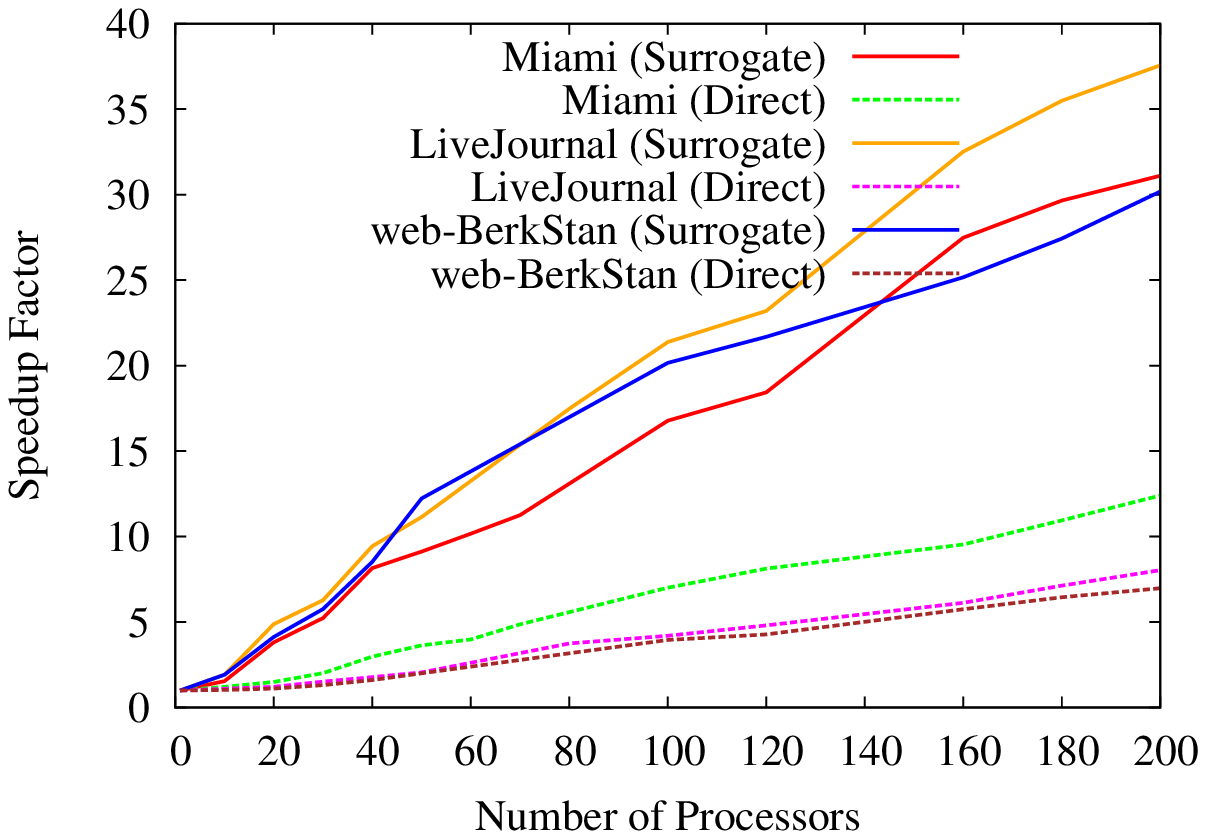}}
\fsc
\caption{Speedup factors of our algorithm with both direct and surrogate approaches.}
\label{fig:strong-space}
\afsc
\end{center}
\end{minipage}
\hfill
\begin{minipage}[t]{.32\textwidth}
\begin{center}
\centerline{\includegraphics[width=1.0\textwidth]{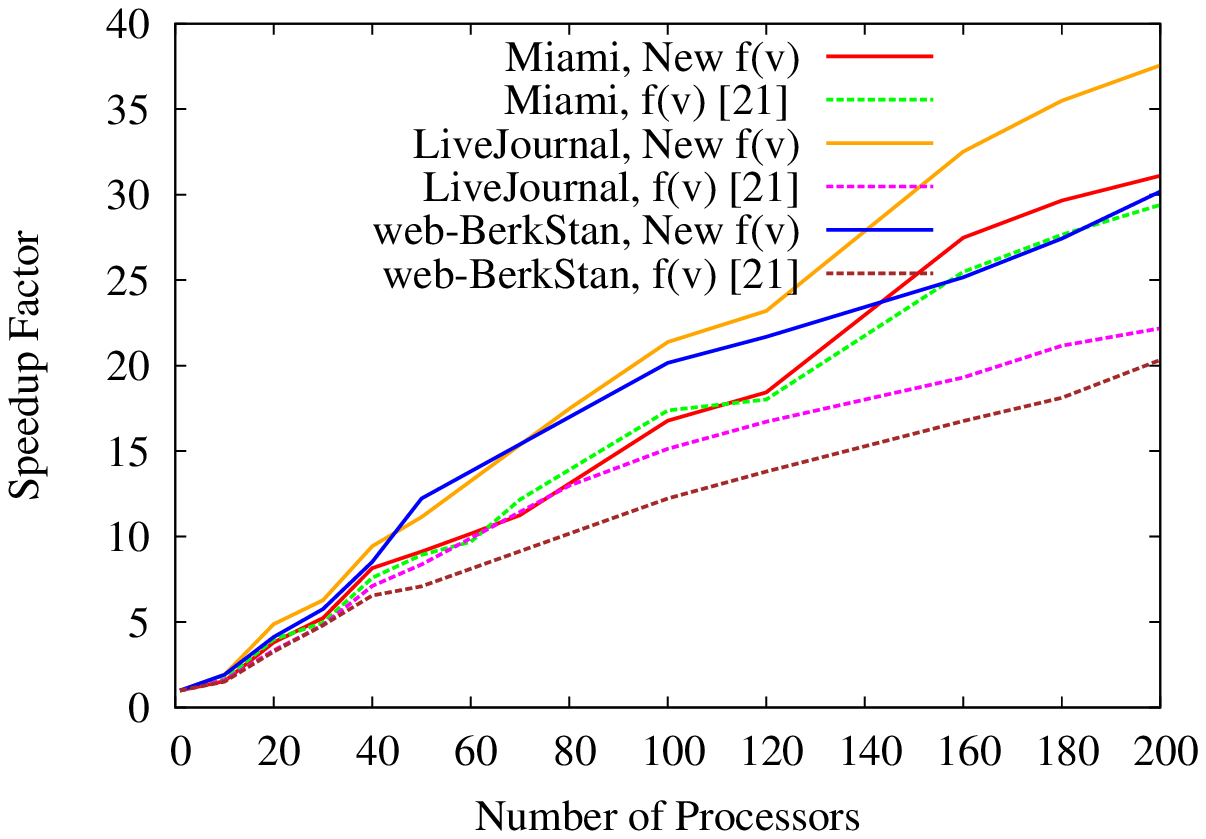}}
\fsc
\caption{Speedup factors of our algorithm with new estimation function and the best function of \cite{Patric-triangle}.}
\label{fig:new_fv}
\afsc
\end{center}
\end{minipage}
\hfill
\begin{minipage}[t]{.32\textwidth}
\begin{center}
\centerline{\includegraphics[width=1.0\textwidth]{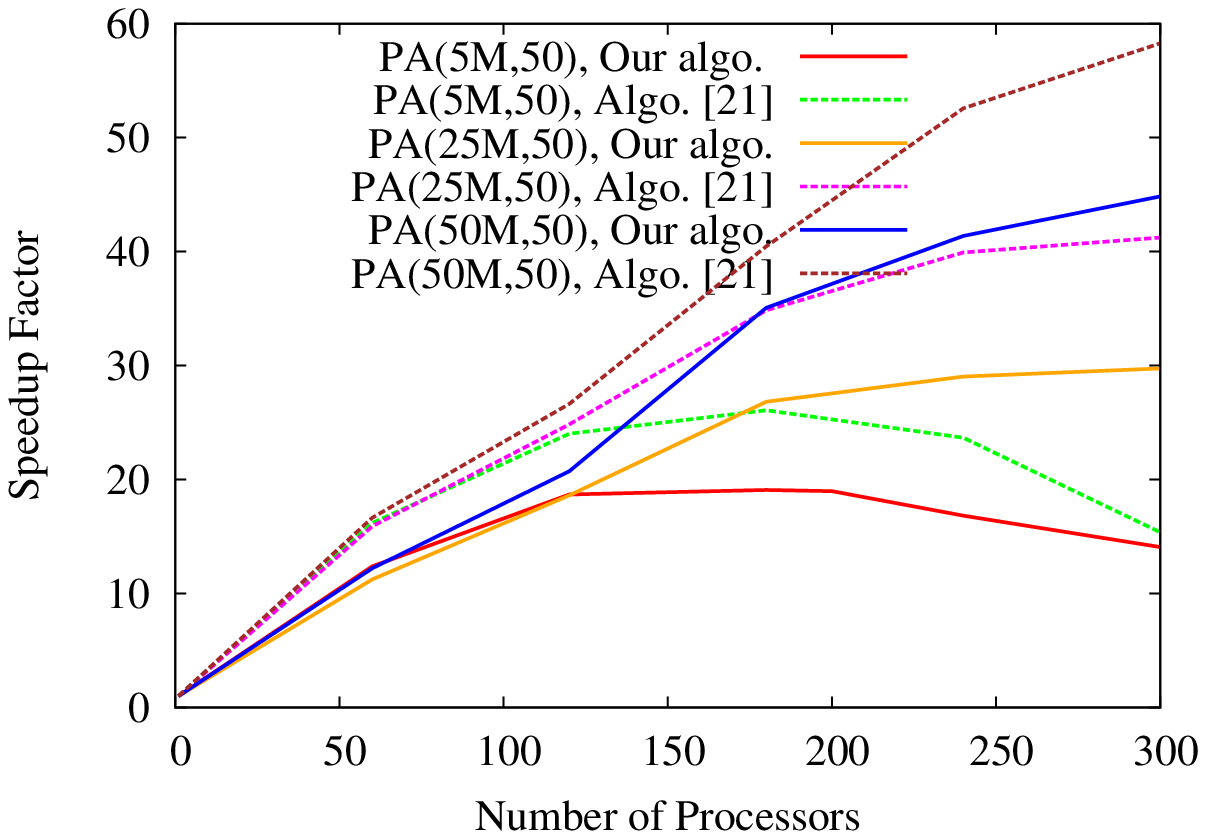}}
\fsc
\caption{Improved scalability of our algorithm with increasing network size.}
\label{fig:space_pscale}
\afsc
\end{center}
\end{minipage}
\hfill
\vspace{-5pt}
\end{figure*}

\subsection{Performance}
\label{sec:space-perf}

In this section, we present the experimental evaluation of the performance our space-efficient parallel algorithm.

\paragraph*{Strong Scaling}

Strong scaling of a parallel algorithm shows how much speedup a parallel algorithm gains as the number of processors increases. Fig. \ref{fig:strong-space} shows strong scaling of our algorithm on Miami, LiveJournal, and web-BerkStan networks with both direct and surrogate approaches. Speedup factors with the surrogate approach are significantly higher than that of the direct approach. The high communication overhead of direct approach due to redundant messages leads to poor speedup whereas the surrogate approach eliminates message redundancy without introducing any computational overhead leading to better performance.  

\paragraph*{Effect of Estimation for f(v)} 


We show the performance of our algorithm with new estimation function (computed in Section \ref{sec:new_fv}), $f(v) = \sum_{ u\in \mathcal{N}_v-N_v}{(\hat{d_v} + \hat{d_u})}$, and the best estimation function $f(v) = \sum_{ u\in N_v}{(\hat{d_v} + \hat{d_u})}$ of \cite{Patric-triangle}. As Fig. \ref{fig:new_fv} shows, our algorithm with new estimation function provides better speedup than that with \cite{Patric-triangle}.  Miami network has a comparable performance with both functions since it has a relatively even degree distribution and both functions provide almost the same estimation. However, for networks with skewness in degrees (LiveJournal and web-BerkStan), our new function estimates the computational cost more precisely and provides significantly better speedup.   

\paragraph*{Scaling with Processors and Network Size} 

We show how our algorithm scales with increasing network size and number of processors, and compare the results with the algorithm in \cite{Patric-triangle}. Our algorithm scales to a higher number of processors when networks grow larger, as shown in Fig. \ref{fig:proc_scale}. This is, in fact, a highly desirable behavior of our parallel algorithm since we need a large number of processors when the network size is large and computation time is high. Now as compared to \cite{Patric-triangle}, the speedup and scaling of our algorithm are a little smaller since our algorithm has a higher communication overhead. However, this difference in scaling is very small and both algorithms perform comparably.


\paragraph*{Memory Scaling} 

We compare the space requirement of our algorithm and \cite{Patric-triangle} with networks with increasing average degrees. For this experiment, we use $PA(10M, d)$ networks with average degree $d$ varying from $10$ to $100$. As shown in Fig. \ref{fig:memory_comp}, our algorithm shows a very linear (and slow) increase of space requirement whereas for \cite{Patric-triangle} the space requirement increases very rapidly. Our algorithm divides the network into non-overlapping partitions and hence has a much smaller space requirement as discussed in Section \ref{sec:background}. For the same reason, space requirement of our algorithm for storing a partition reduces rapidly with additional processors, as shown in Fig. \ref{fig:memory_scale}. 

\begin{figure*}[!tbh]
\hfill
\begin{minipage}[t]{.32\textwidth}
\begin{center}
\centerline{\includegraphics[width=1.0\textwidth]{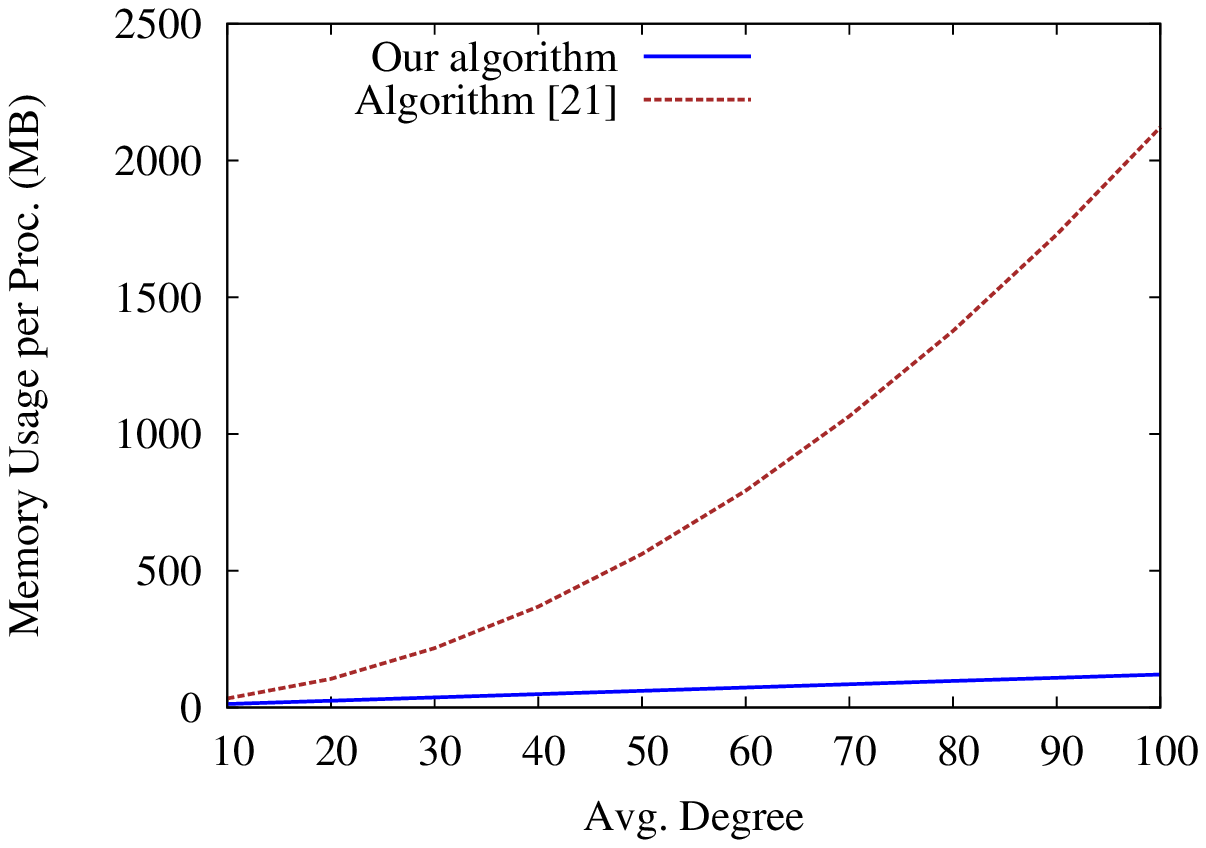}}
\fsc
\caption{Increase of space requirement of our algorithm and algorithm \cite{Patric-triangle} with increasing average degree of networks.}
\label{fig:memory_comp}
\afsc
\end{center}
\end{minipage}
\hfill
\begin{minipage}[t]{.32\textwidth}
\begin{center}
\centerline{\includegraphics[width=1.0\textwidth]{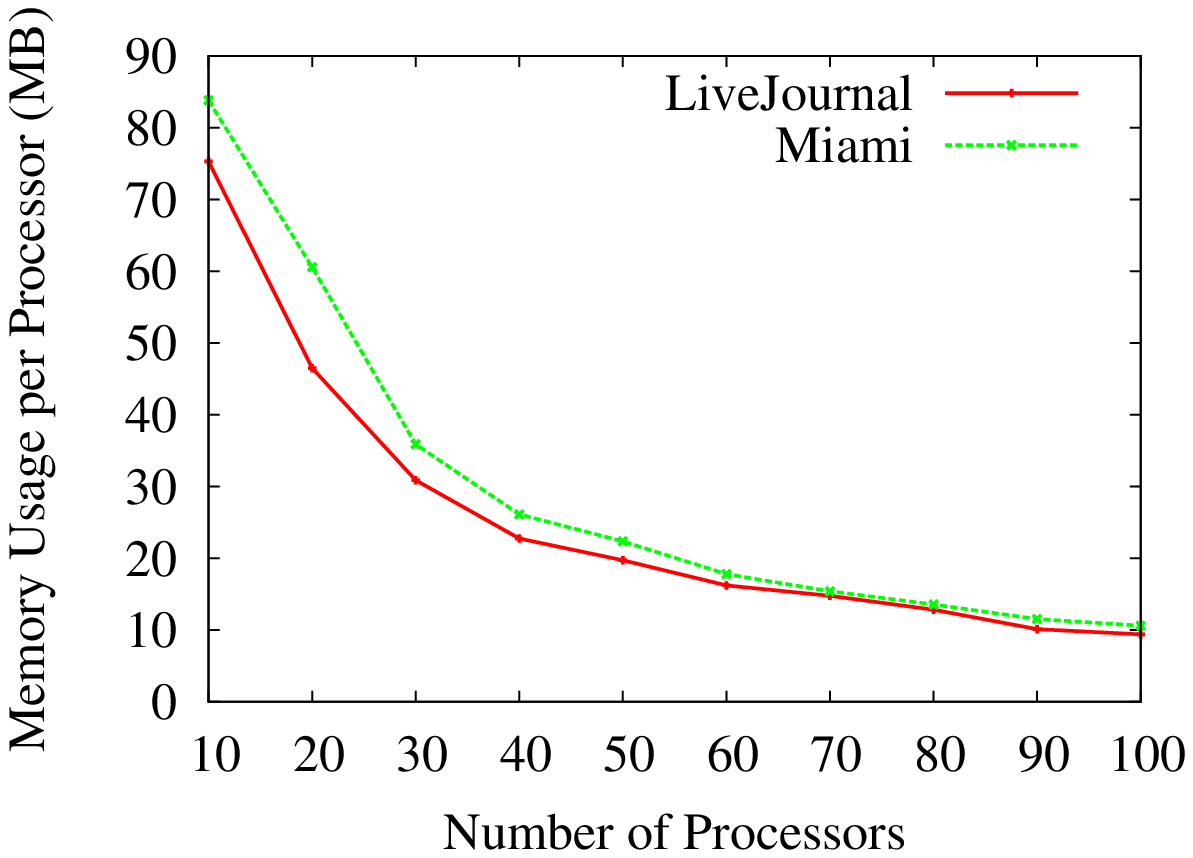}}
\fsc
\caption{Memory scalability of our algorithm with increasing number of processors for Miami and LiveJournal networks.}
\label{fig:memory_scale}
\afsc
\end{center}
\end{minipage}
\hfill
\begin{minipage}[t]{.32\textwidth}
\begin{center}
\centerline{\includegraphics[width=1.0\textwidth]{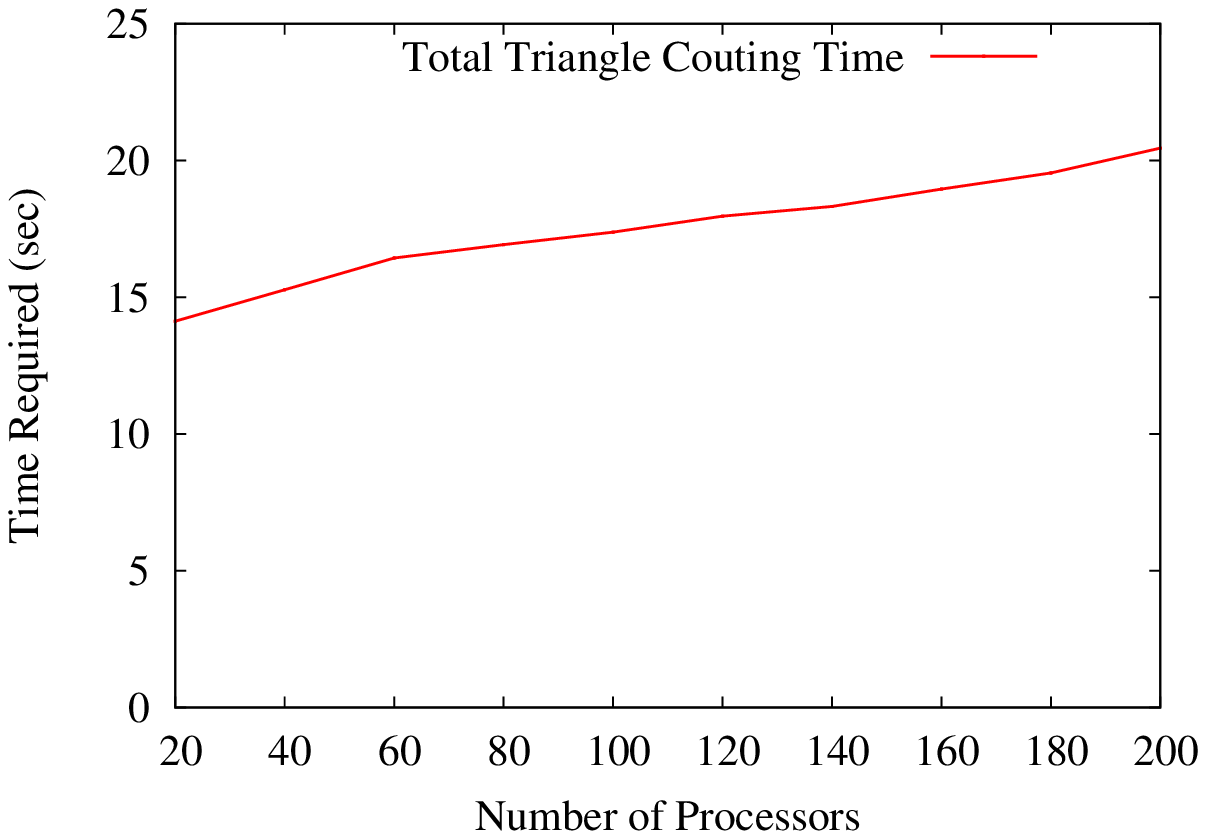}}
\fsc
\caption{Weak scaling of our algorithm. The experiment is performed on PA($P/10*1M, 50$) networks where $P$ is the number of processors used.}
\label{fig:weak-space}
\afsc
\end{center}
\end{minipage}
\hfill
\vspace{-5pt}
\end{figure*}

\paragraph*{Comparison of Runtime with Previous Algorithms}

We present a comparison of runtime of our algorithm with the algorithm in \cite{Patric-triangle} in Table \ref{table:reusage}. Since our algorithm exchanges messages for counting triangles, it has a higher runtime than \cite{Patric-triangle}. Runtime with the direct approach is relatively high due to message redundancy. However, our algorithm with surrogate approach improves the runtime quite significantly, and the performance is quite comparable to \cite{Patric-triangle}.  

\begin{table}[!hbt]
\centering
\caption{Runtime performance of our algorithm and the algorithm in \cite{Patric-triangle}. We used 200 processors for this experiment.}
\label{table:reusage}
\begin{tabular}{|l|l|l|l|l|} \hline
\multirow{2}{*}{ \textbf{Networks}} &  \multicolumn{3}{|c|}{\textbf{Runtime }} & \multirow{2}{*}{\textbf{Triangles}} \\ 
\cline{2-4}
 & \cite{Patric-triangle} & Direct & Surrogate& \\ \hline
	web-BerkStan  &    $0.10$s  & $3.8$s  & $0.14$s & $65$M    \\ \hline 
	Miami         &     $0.6$s    & $4.79$s & $0.79$s & $332$M    \\ \hline 
	LiveJournal   &   $0.8$s     & $5.12$s  & $1.24$s &$286$M   \\ \hline 
	Twitter   &   $9.4$m    &  $35.49$m & $12.33$m  &$34.8$B  \\ \hline 
	PA(1B, 20)         &     $15.5$m    & $78.96$m & $20.77$m &$0.403$M  \\ \hline 
    \end{tabular}
\end{table}

\paragraph*{Weak Scaling} 

Weak scaling of a parallel algorithm measures its ability to maintain constant computation time when the problem size grows proportionally with the number of processors. The weak scaling of our algorithm is shown in Fig. \ref{fig:weak-space}. Although the problem size per processor remains same, the addition of processors causes the overhead for exchanging messages to increase. Thus with increasing number of processors, runtime of our algorithm increases slowly. Since the change in runtime is not very drastic, rather slow, the weak scaling of the algorithm is good.

\section{A fast parallel algorithm with dynamic load balancing}
\label{sec:dynamic}

In this section, we present our parallel algorithm for counting triangles with an efficient dynamic load balancing scheme. First, we provide an overview of the algorithm, and then a detailed description follows.

\subsection{Overview of the Algorithm}
\label{sec:dynamic-ov}

We assume that each computing node has enough memory to store the whole network. The computation of counting triangles in the network is distributed among processors. We refer the computation assigned to and performed by a processor as a task. For the convenience of future discussion, we present the following definitions related to computing tasks.

\begin{definition} \label{dfn:task}
\textbf{Task:} Given a graph $G=(V, E)$, a task denoted by $\left\langle v, t \right\rangle $, refers to counting triangles incident on nodes $v \in \{v, v+1, \dots, v+t-1 \} \subseteq V$. The task referring to counting triangles in the whole network is $\left\langle 0, n \right\rangle $.  
\end{definition}

\begin{definition} \label{dfn:atomic_task}
\textbf{An atomic task:} A task $\left\langle v, 1 \right\rangle $ referring to counting triangles incident on a single node $v$ is an atomic task. An atomic task can not be further divided. 
\end{definition}

\begin{definition} \label{dfn:task_size}
\textbf{Task size:} Let, $f: V \rightarrow \mathcal{R}$ be a cost function such that $f(v)$ denotes some measure of the cost for counting triangles on node $v$. We define the size $S(v,t)$ of a task $\left\langle v, t \right\rangle$  as follows.
\begin{eqnarray*}
S(v,t) = \sum_{i=0}^{t-1}{f(v+i)}.
\end{eqnarray*}
\end{definition}
We consider the cost functions $f(v) =1$ and $f(v) =d_v$ since those are known for all $v\in V$ and have no computational overhead. Using a computationally expensive function for  representing the cost for counting triangles might lead to poor performance of the algorithm.



Now in a static load balancing scheme, each processor works on a pre-computed partition. Since the partitioning is based on estimated computing cost which might not equal to the actual computing cost, some processors will remain idle after finishing computation ahead of others. Our algorithm employs a dynamic load balancing scheme to reduce idle time of processors leading to improved performance. The algorithm divides the total computation into several tasks and assign them dynamically. How and when to assign a task require communication among processors. The scheme for Communication and decision about task granularity are crucial to the performance of our algorithm.

In the following subsection, we describe an efficient algorithm for dynamic load balancing. 

\subsection{An Efficient Dynamic Load Balancing Scheme}
\label{sec:dynamiclb}

We design a dynamic load balancing scheme with a dedicated processor for coordinating balancing decisions. We distinguish this processor as the \textit{coordinator} and the rest as \textit{workers}. The \textit{coordinator} assigns tasks, receives notifications and re-assigns tasks to idle workers, and  \textit{workers} are responsible for actually performing \textit{tasks}. At the beginning, each worker is assigned an initial task. Once any worker $i$ completes its current task, it sends a request to \textit{coordinator} for an additional task. From the available un-assigned tasks, \textit{coordinator} assigns a new task to worker $i$. 

The coordinate may divide the computation into tasks of equal size and assign them dynamically. However, the size of tasks is a crucial determinant of the performance of the algorithm. Assume time required by some worker to compute the last completed task is $q$. The amount of time a worker remains idle, denoted by a continuous random variable $X$, can be assumed to be uniformly distributed over the interval $[0,q]$, i.e., $X \sim U(0,q)$. Since $E[X]=q/2$, a worker remains idle for $q/2$ amount of time on average. If the size $S(v,t)$ of tasks $\left\langle v, t \right\rangle $ is large, time $q$ required to complete the last task becomes large, and consequently, idle time $q/2$ also grows large. In contrast, if $S(v,t)$ decreases, the idle time is expected to decrease. However, if $S(v,t)$ is very small, total number of tasks becomes large, which increases communication overhead for task requests and re-assignments. 

Therefore, instead of keeping the size of tasks $S(v,t)$ constant throughout the execution, our algorithm adjusts $S(v,t)$ dynamically, initially assigning large tasks and then gradually decreasing them. 
In particular, initially half of the total computation $\left\langle 0, n \right\rangle $ is assigned among the workers in tasks of almost equal sizes. That is, a total of $\left\langle 0, t' \right\rangle $ task, such that $S(0,t')=\frac{1}{2}S(0,n)$, is assigned initially, and the remaining computations $\left\langle t', n-t' \right\rangle$ are assigned dynamically with the granularity of tasks decreasing gradually. Next, we describe the steps of our dynamic load balancing scheme in detail.  

\textbf{Initial Assignment.} 
The set of $(P-1)$ initial tasks corresponds to counting triangles on nodes $v \in \{0, 1, \dots, t'-1\}$ such that $S(0,t') \approx S(t',n-t')$. Thus we need to find node $t'$ which divides the set of nodes $V$ into two disjoint subsets in such a way that $\sum_{v=0}^{t'-1}{f(v)}\approx \sum_{v=t'}^{n-1}{f(v)}$, given $f(v)$ for each $v\in V$. Now if we compute sequentially, it takes $O(n)$ time to perform the above computations. However, we observe that a parallel algorithm for computing balanced partitions of $V$ proposed in \cite{Patric-triangle} can be used to perform the above computation which takes $O(n/P+ \log P)$ time. Once $t'$ is determined, the task $\left\langle 0, t' \right\rangle $ is divided into $(P-1)$ tasks $\left\langle v, t \right\rangle$, one for each worker, in almost equal sizes. 
\begin{equation} \label{eqn:init_task}
S(v,t) = \frac{1}{P-1}\sum_{v\in 0}^{t'-1}{f(v)}.
\end{equation}
That is, set of nodes $\{0, 1, \dots, t'-1\}$ is divided into $(P-1)$ subsets such that for each subset $\{v, v+1, \dots, t-1\}$,  $\sum_{i=0}^{t-1}{f(v+i)} \approx \frac{1}{P-1}\sum_{v\in 0}^{t'-1}{f(v)}$. This computation can also be done using the parallel algorithm \cite{Patric-triangle} mentioned above.

Note that all $P$ processors work in parallel to determine initial tasks. Since the initial assignment is deterministic, workers pick their respective tasks $\left\langle v, t \right\rangle$ without involving the coordinator. 

\textbf{Dynamic Re-assignment.} 
Once any worker completes its current task and becomes idle, the \textit{coordinator} assigns it a new task dynamically. This re-assignment is done in the following steps. 
\begin{itemize}
\item The coordinator divides the un-assigned computations $\left\langle t', n-t' \right\rangle$ into several tasks and stores them in a queue $W$. How the coordinator decides on the size $S(v,t)$ of each task  $\left\langle v, t \right\rangle$ will be described shortly. 
\item When any worker $i$ finishes its current task and becomes idle, it sends a task request $\left\langle i \right\rangle $ to the coordinator.
\item If $W \neq \emptyset$, the coordinator picks a task  $\left\langle v, t \right\rangle \in W$, and assigns it to worker $i$.
\end{itemize}

Our algorithm decreases the size $S(v,t)$ of each dynamically assigned tasks gradually for the reasons discussed at the beginning of this subsection. Let, $V'$ be the set of nodes remaining to be assigned as tasks. Since at every new assignment $V'$ decreases our algorithm uses $V'$ to dynamically adjust task sizes. This is done using the following equation. 
\begin{equation} \label{eqn:task_adj}
S(v,t) = \frac{1}{P-1}\sum_{v \in V'}{f(v)}.
\end{equation}
Note that the size $S(v,t)$ of a dynamically assigned task $\left\langle v, t \right\rangle $ decreases at every new assignment. By the definition of atomic task (in definition \ref{dfn:atomic_task}) we have a finite number of tasks. When the coordinator has no more unassigned tasks, i.e., $W = \emptyset$, it sends a special termination message $\left\langle terminate \right\rangle$ to the requesting worker. Once the coordinator completes sending termination messages to all workers, it aggregates counts of triangles from all workers, and the algorithm terminates.

\subsection{Counting Triangles}
\label{sec:dynamic-tc}

Once a processor $i$ has an assigned task $\left\langle v, t \right\rangle $, it uses the algorithm presented in Fig. \ref{algo:tcount} to count the triangles incident on nodes $v \in \{v, v+1, \dots, v+t-1 \}$.

\begin{figure}[!tbh]
\begin{center}
\fbox{
\begin{minipage}[c] {0.7\linewidth}
\begin{algorithmic}[1]
\STATE {Procedure \textsc{CountTriangles}$(\left\langle v, t \right\rangle):$}
	\STATE {$T \leftarrow 0$}\hspace{0.1in} //$T$ is the count of triangles
\FOR {$v \in \{v, v+1, \dots, v+t-1 \}$}
	\FOR {$u \in N_v$}
		\STATE $S \leftarrow N_v \cap N_u$
		\STATE $T \leftarrow T+|S|$
	\ENDFOR
\ENDFOR
\RETURN  $T$	
\end{algorithmic}
\end{minipage}
}
\end{center}
\fsc
\caption{A procedure executed by processor $i$ to count triangles corresponding to the task $\left\langle v, t \right\rangle $.}
\label{algo:tcount}
\afsc
\end{figure}

The complete pseudocode of our algorithm for counting triangles with an efficient dynamic load balancing scheme is presented in Fig. \ref{algo:dynamic}.

\begin{figure}[!ht]
\begin{center}
\fbox{
\begin{minipage}[c] {0.94\linewidth}
\begin{algorithmic}[1]
\STATE {\textbf{All processors initially do the following:}}
\STATE {\textbf{Determine} \textit{initial} tasks\hspace*{.05in}(see discussion of Eqn. \ref{eqn:init_task}) }
\STATE 

\STATE {\textbf{The \textit{coordinator} does the following:}}
\STATE {$W\leftarrow \emptyset$}
\FOR {all remaining tasks $\left\langle v, t \right\rangle $ } 
		\STATE {\textsc{EnQueue} ($W$, $\left\langle v, t \right\rangle $ ) }
\ENDFOR

\WHILE {$W$ is not $\emptyset$}
	\STATE {\textbf{Receive} \textit{task} requests $\left\langle i \right\rangle $ }	
	\STATE {$\left\langle v, t \right\rangle \leftarrow$ \textsc{DeQueue} ($W$)}
	\STATE {\textbf{Send} message $\left\langle v, t \right\rangle$ to worker $i$} 
\ENDWHILE
\STATE {\textbf{Send} $\left\langle terminate \right\rangle$ to proc. $i$ for requests $\left\langle i \right\rangle $ }
\STATE
 
\STATE {\textbf{Each \textit{worker} $i$ does the following:}}
\STATE $T_i \leftarrow 0$
\STATE {$T_i \leftarrow T_i + \textsc{CountTriangles}(v, t)$ \hspace*{.01in} //for \textit{initial} task}
\WHILE {\textit{worker} $i$ is \textit{idle}}
	\STATE {\textbf{Send} message  $\left\langle i \right\rangle $ to \textit{coordinator}}
	\STATE {\textbf{Receive} message  $M$ from \textit{coordinator}}
	\IF {$M$ is $\left\langle terminate \right\rangle$}
		\STATE {\textbf{Stop} execution}
	\ELSIF {$M$ is a \textit{task} $\left\langle v, t \right\rangle $}
		\STATE $T_i \leftarrow T_i + \textsc{CountTriangles}(v, t)$
	\ENDIF
\ENDWHILE
\STATE
\STATE \textsc{MpiBarrier}
\STATE Find Sum $T \leftarrow \sum_i{T_i}$ using $\textsc{MpiReduce}$
\RETURN  $T$
\end{algorithmic}
\end{minipage}
}
\end{center}
\fsc
\caption{An algorithm for counting triangles with dynamic load balancing.  }
\label{algo:dynamic}
\afsc
\end{figure}

\subsection{Correctness of the Algorithm}
\label{sec:dynamic-correctness}

We establish the correctness of our algorithm as follows. Consider a triangle $(x_1, x_2, x_3)$ with $x_1 \prec x_2 \prec x_3$, without the loss of generality. Now, the triangle is counted only when $x_1 \in \{v, v+1, \dots, v+t-1 \}$ for some task $\left\langle v, t \right\rangle $. The triangle is never counted again since $x_1 \notin N_{x_2}$ and $x_1,x_2 \notin N_{x_3}$ by the construction of $N_x$ (Line 1-3 in Fig. \ref{algo:serial}).

\subsection{Performance}
\label{sec:dynamic-perf}

In this section, we present the experimental evaluation of our parallel  algorithm for counting triangles with dynamic load balancing. 

\begin{figure}[!ht]
\centering
\includegraphics[scale=0.45]{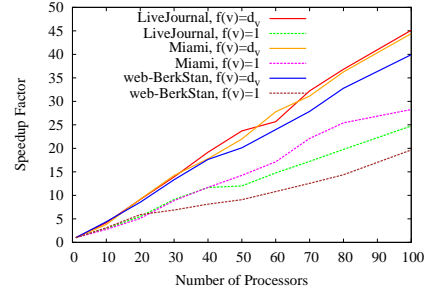}
\fsc
\caption{Speedup factors of our algorithm on Miami, LiveJournal and web-BerkStan networks with both $f(v)=1$  and $f(v) = d_v$ cost functions. }
\label{fig:dlb_speed_comp}
\afsc
\end{figure}

\begin{figure*}[!htb]
\hfill
\begin{minipage}[t]{.32\textwidth}
\begin{center}
\centerline{\includegraphics[width=1.0\textwidth]{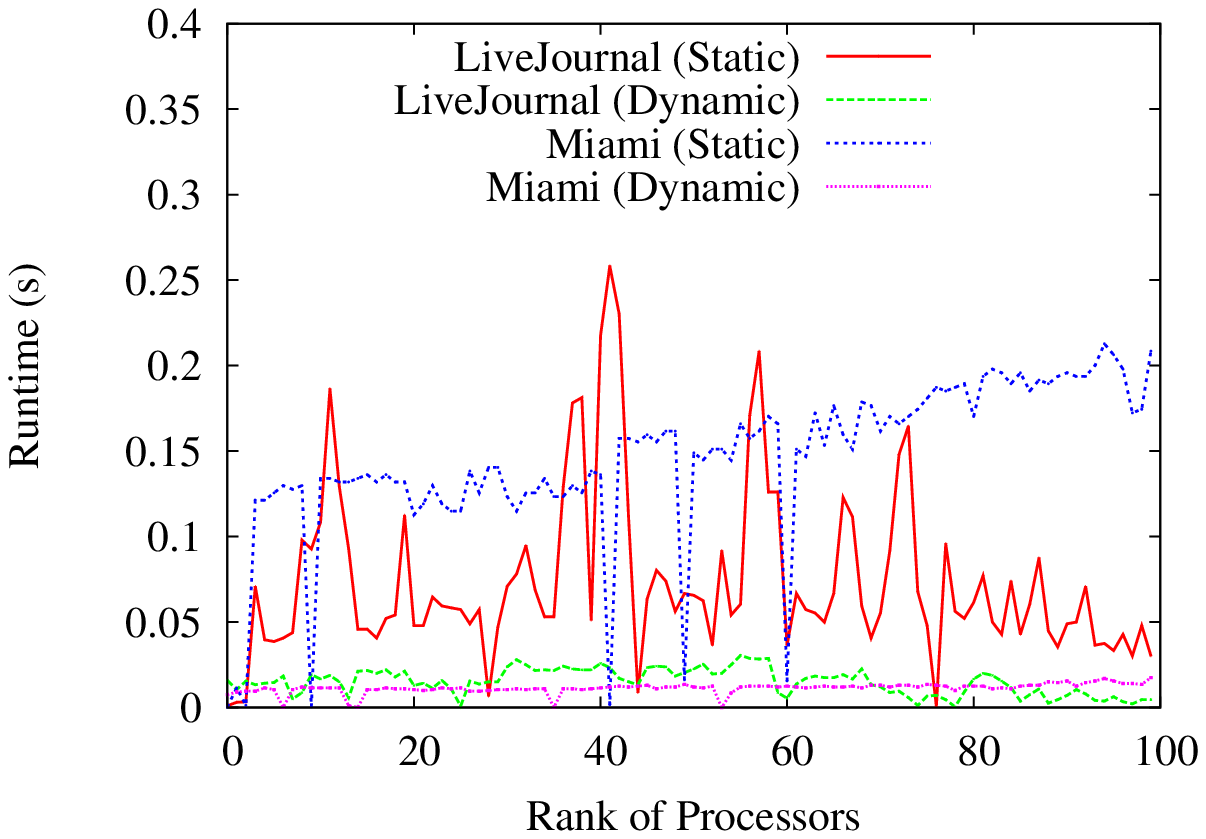}}
\fsc
\caption{Idle time of processors with both static and dynamic task granularity. When task granularity is adjusted dynamically idle times decrease significantly leading to smaller runtime.}
\label{fig:idle_vs_total}
\afsc
\end{center}
\end{minipage}
\hfill
\begin{minipage}[t]{.32\textwidth}
\begin{center}
\centerline{\includegraphics[width=1.0\textwidth]{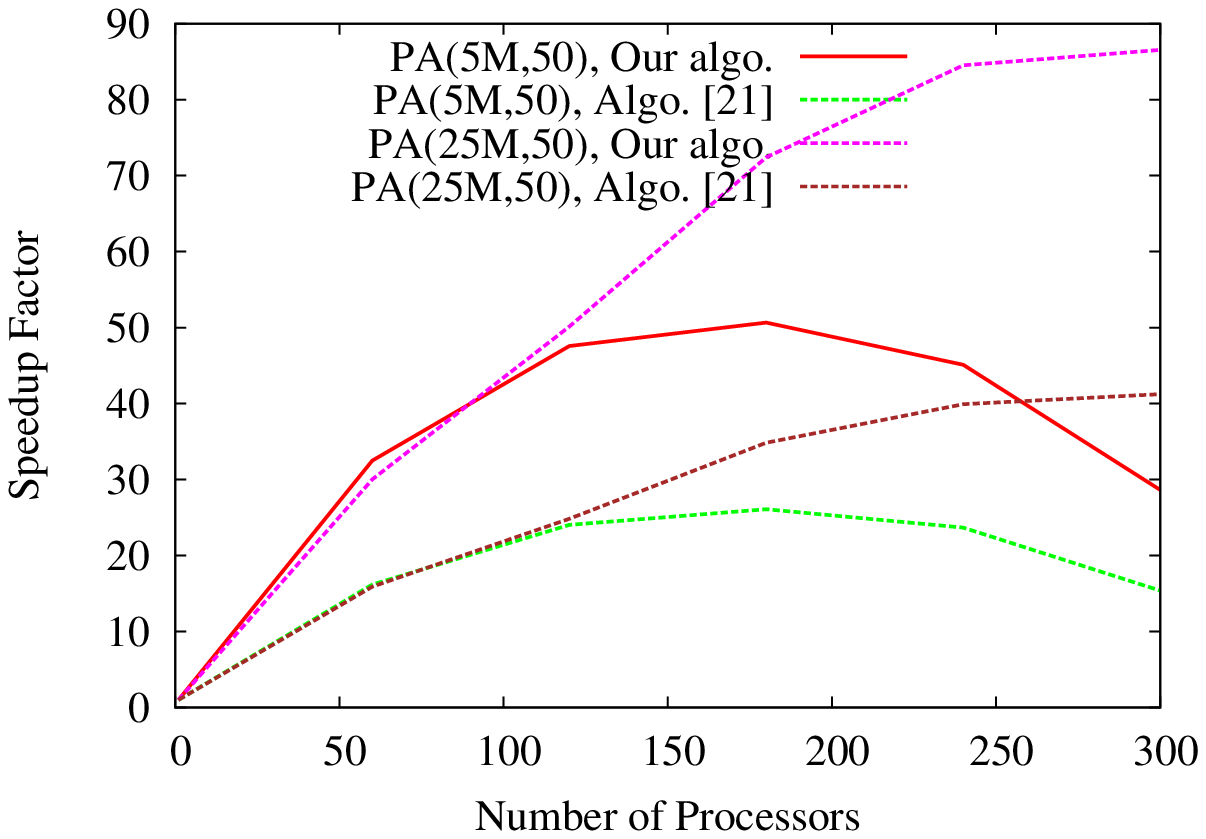}}
\fsc
\caption{Improved scalability of our algorithm with increasing network size. Further, our algorithm achieves higher speedups than \cite{Patric-triangle}.}
\label{fig:proc_scale}
\afsc
\end{center}
\end{minipage}
\hfill
\begin{minipage}[t]{.32\textwidth}
\begin{center}
\centerline{\includegraphics[width=1.0\textwidth]{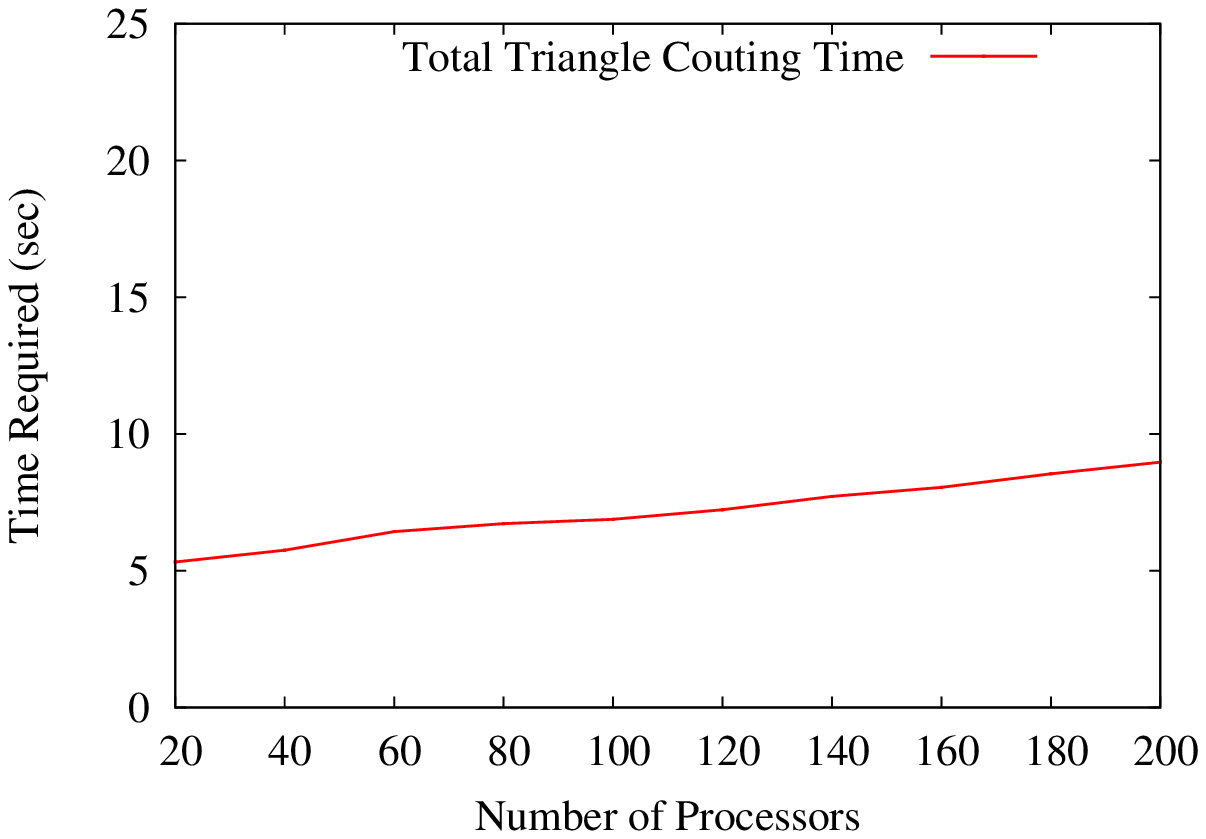}}
\fsc
\caption{Weak scaling of our algorithm. A very slowly increasing runtime suggests a good weak scaling of the algorithm.}
\label{fig:weak}
\afsc
\end{center}
\end{minipage}
\hfill
\vspace{-5pt}
\end{figure*}

\paragraph*{Strong Scaling} 

We present the strong scaling of our algorithm on Miami, LiveJournal, and web-BerkStan networks with both cost functions $f(v)=1$ and $f(v) = d_v$ in Fig. \ref{fig:dlb_speed_comp}.  Speedup factors are significantly higher with the function $f(v) = d_v$ than with $f(v) = 1$. The function $f(v)=1$ refers to equal cost of counting triangles for all nodes whereas the function $f(v) = d_v$ relates the cost to the degree of $v$. Distributing tasks based on the sum of degrees of nodes (Eqn. \ref{eqn:init_task} and \ref{eqn:task_adj}) reduces the effect of skewness of degrees and makes tasks more balanced leading to higher speedups. Our next experiments will be based on cost function $f(v) = d_v$.

\paragraph*{Comparison with Previous Algorithms} 
We compare the runtime of our parallel algorithm with the algorithm in \cite{Patric-triangle} on a number of real and artificial networks. As shown in Table \ref{table:dynamic_patric}, our algorithm is is more than 2 times faster than \cite{Patric-triangle} for all these networks.
The algorithm in \cite{Patric-triangle} is based on static partitioning whereas our algorithm employs a dynamic load balancing scheme to reduce idle time of processors leading to improved performance.

\begin{table}[!hb]
\centering
\caption{Runtime performance of our algorithm and algorithm \cite{Patric-triangle}.}
\label{table:dynamic_patric}
\begin{tabular}{|l|l|l|l|} \hline
\multirow{2}{*}{ \textbf{Networks}} &  \multicolumn{2}{|c|}{\textbf{Runtime }} & \multirow{2}{*}{\textbf{Triangles}} \\
\cline{2-3}
 & \cite{Patric-triangle} & Our algo. &  \\ \hline
	web-BerkStan  &    $0.10$s  & $0.041$s  & $65$M    \\ \hline 
	LiveJournal   &    $0.8$s    & $0.384$s  & $286$M   \\ \hline 
	Miami         &    $0.6$s     & $0.301$s & $332$M    \\ \hline 
	PA(20M, 50)         &    $11.85$s     & $5.241$s &  $0.028$M  \\ \hline 
    \end{tabular}
\end{table}

\paragraph*{Effect of Dynamic Adjustment of Task Granularity}
We show how the granularity of tasks affects idle time of worker processors for Miami and LiveJournal networks. As Fig. \ref{fig:idle_vs_total} shows, with tasks of static size, some processors have very large idle times. However, when task granularity is dynamically adjusted,  idle times of processors become very small leading to balanced load among processors. This consequently improves the runtime performance of the algorithm. 

\paragraph*{Scaling with Processors and Network Size} 

Our algorithm scales to a higher number of processors with increasing size of networks, as shown in Fig. \ref{fig:proc_scale}. Scaling of our algorithm with number of processors is very comparable to that of \cite{Patric-triangle}. However, our algorithm achieves significantly higher speedup factors than \cite{Patric-triangle}.


\paragraph*{Weak Scaling}
The weak scaling of our algorithm is shown in Fig. \ref{fig:weak}. With the addition of processors, communication overhead increases since idle workers exchange messages with the coordinator for new tasks. However, since the overhead for requesting and assigning tasks is very small, the increase of runtime with additional processors is very slow. Thus, the weak scaling of our algorithm is very good.

\section{Conclusion}
\label{sec:conc}

We present two parallel algorithms for counting triangles in networks with large degrees. Our space-efficient algorithm works on networks that do not fit into the main memory of a single computing machine. The algorithm partitions the network into non-overlapping subgraphs and reduces the memory requirement significantly leading to the ability to work with even larger networks. This algorithm is not only useful for networks with large degrees, it is equally useful for networks with small degrees when it comes to working with massive networks. We believe that for emerging massive networks, this algorithm will prove very useful. When the main memory of each computing machine is large enough to store the whole network, our parallel algorithm with dynamic load balancing can be used for faster analysis of the network.

\small
\bibliographystyle{./IEEEtran}
\bibliography{./IEEEabrv,./triangle}
\end{document}